\documentclass[false]{llncs}
    \usepackage{graphicx,tipa}
    \usepackage[ruled]{algorithm}
    \usepackage[noend]{algpseudocode}
    \algdef{SE}[DOWHILE]{Do}{doWhile}{\algorithmicdo}[1]{\algorithmicwhile\ #1}
    \usepackage{comment}
    \usepackage{booktabs} 

\usepackage{color}
\usepackage{subfig}
\usepackage{epsfig}
\usepackage{epstopdf}
\usepackage{amsfonts}
\usepackage{latexsym}
\usepackage{amssymb}
\usepackage{amsmath}
    
    \newtheorem{claims}[theorem]{Claim}
    \newtheorem{observation}[theorem]{Observation}

    \newcommand{\reals}{\mathbb{R}}

    \newcommand{\sinn}[1]{\sin \left({#1}\right)}
    \newcommand{\coss}[1]{\cos \left({#1}\right)}

    \newcommand{\acoss}[1]{\cos^{-1} \left({#1}\right)}
    \newcommand{\atann}[1]{\tan^{-1} \left({#1}\right)}

    \newcommand{\pair}[2]{\left({#1},\ {#2}\right)}

    \newcommand{\BO}[1]{\mathcal{O} \left({#1}\right)}

    \newcommand{\diff}[2]{\frac{d{#1}}{d{#2}}}
    \newcommand{\abs}[1]{\left|{#1}\right|}
    \newcommand{\ignore}[1]{}

    \newcommand{\UC}{\mathcal{U}}
    \newcommand{\IT}{\mathcal{T}}
    \newcommand{\IG}{\mathcal{G}}
    
    \newcommand{\IS}{\mathcal{S}}
    \newcommand{\IC}{\mathcal{C}}
    \newcommand{\ID}{\mathcal{D}}

    \newcommand{\eps}{\epsilon}
    \newcommand{\sclass}{\mathbb{S}}
    \newcommand{\symclass}{\mathbb{S}_{sym}}

    \newcommand{\Alg}{\mathcal{A}}

    \newcommand{\pe}[1]{\textsc{PEvac}$_{#1}$}

    \newcommand{\includeFig}[4]{\begin{figure}[htb!] \begin{center} \includegraphics[{#1}]{#3}\caption{\label{#2}#4} \end{center} \end{figure}} 

\begin{document}

\title{Priority Evacuation from a Disk \\
Using Mobile Robots
\thanks{This is the full version of the paper with the same title which will appear in the proceedings of the 25th International Colloquium on Structural Information and Communication Complexity, June 18-21,  2018, Ma'ale HaHamisha, Israel.}
}

\author{
    J. Czyzowicz\inst{1}\inst{7}
    \and
    K. Georgiou\inst{2}\inst{7}
    \and
    R. Killick\inst{3}\inst{8}
    \and
    E. Kranakis\inst{3}\inst{7}
    \and
    D. Krizanc\inst{4}
    \and
    L. Narayanan\inst{5}\inst{7}
    \and
    J. Opatrny\inst{5}\inst{7}
    \and
    S. Shende\inst{6}
    }
    
    \institute{
    D\'{e}partemant d'informatique, Universit\'{e} du Qu\'{e}bec en Outaouais,  Gatineau, Canada.
    \and
    Department of Mathematics, Ryerson University, Toronto, Canada
    \and
    School of Computer Science, Carleton University, Ottawa, Ontario, Canada.
    \and
    Department of Mathematics \& Comp. Sci., Wesleyan University, Middletown CT, USA
    \and
    Department of Comp. Sci. and Software Eng., Concordia University, Montreal, QC,  Canada
    \and
    Department of Computer Science, Rutgers University, Camden, USA
    \and
    Research supported in part by NSERC Discovery grant.
    \and
    Research supported by the Ontario Graduate Scholarship.
    }

\maketitle
\begin{abstract}

We introduce and study a new search-type problem with ($n+1$)-robots on a disk. The searchers (robots) all start from the center of the disk, have unit speed, and can communicate wirelessly. The goal is for a distinguished robot (the queen) to reach and evacuate from an exit that is hidden on the perimeter of the disk in as little time as possible. The remaining $n$ robots (servants) are there to facilitate the queen's objective and are not required to reach the hidden exit. We provide upper and lower bounds for the time required to evacuate the queen from a unit disk. Namely, we propose an algorithm specifying the trajectories of the robots which guarantees evacuation of the queen in time always better than $2 + 4(\sqrt{2}-1) \frac{\pi}{n}$ for $n \geq 4$ servants. We also demonstrate that for $n \geq 4$ servants the queen cannot be evacuated in time less than $2+\frac{\pi}{n}+\frac{2}{n^2}$.

\end{abstract}




\section{Introduction}

A fundamental research topic in mathematics and computer science concerns search, whereby a group of mobile robots need to collectively explore an environment in order to find a hidden target. In the scenarios considered so far, the goal was to optimize the time when the first searcher reaches the target position. More recently, researchers studied the evacuation problem
in which it is required to minimize the time of arrival to the target position of the last mobile robot in the group. 
In the work done on search so far, all robots are generally assumed to have exactly the same capabilities. However, it is quite natural to consider collaborative tasks in which the participant robots have different capabilities. For example, robots may have different maximum speeds, or have different communication capabilities. Robots with different speeds have been studied in the context of rendezvous \cite{feinerman2017fast} and evacuation \cite{lamprou2016fast}. In the context of search, a natural situation may be that only 
one of the robots has the capability to address an urgent need at the target, for example, performing an emergency procedure, or closing a breach in the perimeter. The remaining robots can help in searching for the target, but their arrival at the target does not accomplish the main purpose of finding the target. Therefore, the collective goal of the robots is to get the special robot to the target as soon as possible.
In this paper, we are interested in such a type of search problem, which grants {\em priority} to a pre-selected participant.
In other words, we assume that the collection of robots contains a leader, known in advance, and as long as the leader does not get to the target position, search is considered incomplete. 


More specifically, in this paper we propose and investigate the {\em priority evacuation} problem, a new form of group search in which a given selected searcher in the group is deemed more important than the rest. This distinguished robot is given priority over all other searchers during the evacuation process in that it should be evacuated as early as possible upon the exit being located by any searcher.


\subsection{Model}

In the priority evacuation, or \pe{n}\ problem, $n+1$ robots (searchers) are placed at the center of a unit disk. There is a target (exit), placed at an unknown location on the boundary of the disk. The target can be discovered by any robot walking over it. A robot that finds the exit instantaneously broadcasts its current position. Among the robots there is a distinguished one called the {\em queen} and the remaining $n$ robots are referred to as {\em servants}. The goal is to minimize the queen's \textit{evacuation time}, i.e. the worst case total time until the queen reaches the target. We assume that all robots, including the queen, may walk using maximum unit speed. We note that the queen may or may not actively participate in the search of the exit.

\subsection{Related work}

Search and exploration have been extensively studied in mathematics and various fields of computer science.  If the environment is not known in advance, search implies exploration, and it usually involves mapping and localizing searchers within the environment \cite{AH00,DKP91,HIKK01,PY}. However, even for the case of a known, simple domain like a line, there have been several interesting studies attempting to optimize the search time. These were initiated with the seminal works of Bellman~\cite{bellman1963optimal} and Beck~\cite{beck1964linear}, in which the authors attempted to minimize the competitive ratio in a stochastic setting. 
After the appearance of \cite{baezayates1993searching}, where a search by a single robot was studied for infinite lines and planes, several other works on linear search followed (cf.  \cite{alpern2002theory})
 and more recently the search by a single searcher was studied for different models, e.g., when the turn cost was considered \cite{demaine2006online}, when a bound on the distance to the target is known in advance \cite{Bose13}, and when the target is moving or for more general linear cost functions \cite{Bose16}.

For the case of a collection of searchers, numerous scenarios have been studied, such as: graph or geometric terrains, known or unknown environments, stationary or mobile targets, etc. (cf. \cite{FT08}). In many papers, the objective is to decide the feasibility of the search or to minimize its search time. 

The evacuation problem from the disk was introduced in \cite{CGGKMP} where two types of robots' communication were studied -- the wireless one and communication by contact (also called face-to-face). The bounds for evacuation of two robots communicating face-to-face were later improved in  \cite{DBLP:conf/ciac/CzyzowiczGKNOV15} and in \cite{Watten2017}.
The case of a disk environment with more than one exit was considered in \cite{DBLP:conf/icdcn/CzyzowiczDGKM16} and \cite{pattanayak2017evacuating}. Other variations included evacuation from environments such as regular triangles and squares 
\cite{DBLP:conf/adhoc-now/CzyzowiczKKNOS15}, the case of two robots having different maximal speeds \cite{lamprou2016fast}, and the evacuation problem when one of the robots is crash or byzantine faulty \cite{georgioudiskfaulty2017}.

Group search and evacuation in the line environment were studied in \cite{SIROCCO16,Groupsearch}. The authors of \cite{Groupsearch} proved, somewhat surprisingly, that having many robots using maximal speed $1$ does not reduce the optimal search time as compared to the search using only a single robot. However, interestingly, \cite{Groupsearch} shows that the same bound for group search (and evacuation) is achieved for two robots having speeds $1$ and $1/3$. For both types of robots' communication scenarios, \cite{SIROCCO16} presents optimal evacuation algorithms for two robots having arbitrary, possibly distinct, maximal speeds in the line environment.



A priority evacuation-type problem has been previously considered in~\cite{GeorgiouKK16,GeorgiouKK17} but with different terminology. Using the jargon of the current paper, an immobile queen is hidden somewhere on the unit disk, and a number of robots try to locate her, and fetch (evacuate) her to an exit which is also hidden. The performance of the evacuation algorithm is measured by the time the queen reaches the exit. 
Apart from these results, and to the best of our knowledge nothing is known about the priority evacuation problem. In this work we provide a general strategy for the case of $n \geq 4$ servants. When there are fewer than $4$ servants more ad hoc strategies must be employed which do not fit with the general framework developed here and they are therefore treated elsewhere \cite{CGKKKNOS}.



\subsection{Results of the paper}

Section~\ref{sec2} introduces nomenclature and notation and discusses preliminaries.  In Section~\ref{sec6} we provide an algorithm that evacuates the queen in time always smaller than $2 + 4(\sqrt{2}-1) \frac{\pi}{n}$ for $n \geq 4$ servants (the exact evacuation times of our algorithm must be calculated numerically). In Section~\ref{sec7} we demonstrate that for $n \geq 4$ servants the queen cannot be evacuated in time less than $1 + \frac 2n \cdot \arccos(-\frac 2n ) + \sqrt{1- \frac 4{n^2}}$, or, asymptotically, $2+\frac \pi n+\frac{2}{n^2}$. These results improve upon naive upper and lower bounds of $2+\frac{2\pi}{n}$ and $2+\frac{\pi}{n+1}$ respectively (see Section~\ref{sec:classes} and \ref{sec7}). A summary of the evacuation times for our algorithm (numerical results) as well as the upper and lower bounds (non-trivial and naive) is provided in Table~\ref{tbl:evac} and in Figure~\ref{fig:asymptotic_evac_times}. 
We conclude the paper in Section~\ref{secconclusion} with a discussion of open problems.
\begin{table}[h] \caption{Evacuation times $\IT$ of the queen using Algorithm~\ref{alg:ns} (numerical results). The upper bound of $2+4(\sqrt{2}-1)\frac{\pi}{n}$ (Theorem~\ref{thm:ub}), and the lower bound of $1+\frac 2n \acoss{\frac{-2}{n}}+\sqrt{1-\frac{4}{n^2}}$ (Theorem~\ref{thm:lb}) are also provided. For comparison, the naive upper bound and lower bound of $2+\frac{2\pi}{n}$ (see Section~\ref{sec:classes}) and $2+\frac{\pi}{n+1}$ (see Section~\ref{sec7}) are included.} \label{tbl:evac}
    \begin{center}
    \begin{tabular}{ |c|c|c|c|c|c| }
        \hline
          & $\IT$ & UB & LB & UB & LB\\
         $n$ & (Alg~\ref{alg:ns}) & (Thm~\ref{thm:ub}) & (Thm~\ref{thm:lb}) & Naive & Naive \\
         \hline
         4 & 3.113 & 3.301 & 2.913 & 3.571 & 2.628\\ 
         5 & 2.905 & 3.041 & 2.709 & 3.257 & 2.524\\ 
         6 & 2.762 & 2.868 & 2.580 & 3.047 & 2.449\\ 
         7 & 2.660 & 2.744 & 2.490 & 2.898 & 2.393\\ 
         8 & 2.582 & 2.651 & 2.424 & 2.785 & 2.349\\
         \hline
    \end{tabular}
    \end{center}
\end{table}
\vspace{-0.5cm}
\includeFig{width=4.5in,keepaspectratio}{fig:asymptotic_evac_times}{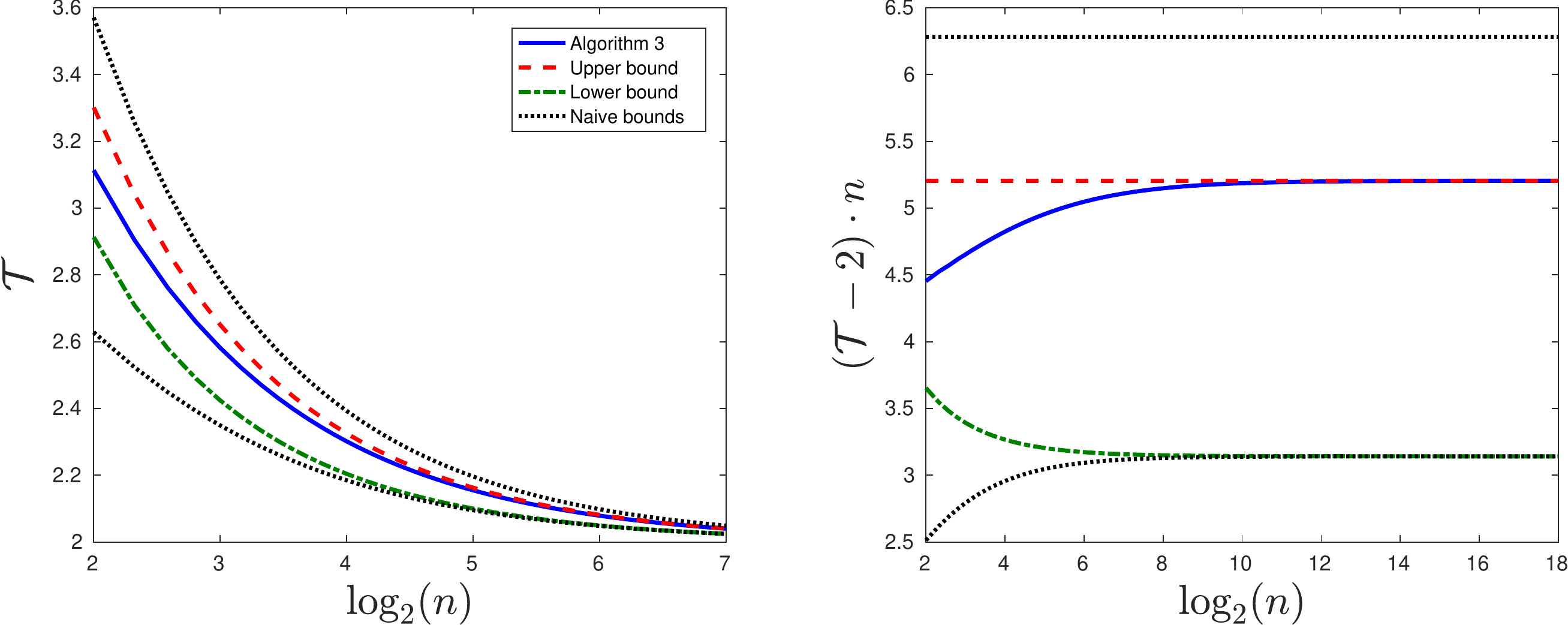}{Evacuation times $
\IT$ of Algorithm~\ref{alg:ns} for $n \in [4,\ 2^{7}]$ (left) and $n \in [4,\ 2^{18}]$ (right). The upper bound of $2+4(\sqrt{2}-1)\frac{\pi}{n}$ (Theorem~\ref{thm:ub}), the lower bound of $1+\frac 2n \acoss{\frac{-2}{n}}+\sqrt{1-\frac{4}{n^2}}$ (Theorem~\ref{thm:lb}) are also provided. For comparison, a naive upper bound and lower bound of $2+\frac{2\pi}{n}$ (see Section~\ref{sec:classes}) and $2+\frac{\pi}{n+1}$ (see Section~\ref{sec7}) are included.}


\section{Notation and Preliminaries}\label{sec2}

In this section we provide some basic notation and terminology and introduce two broad classes of evacuation algorithms.

\subsection{Notation}

We denote by $\UC$ the unit circle in $\reals^2$ centered at the origin $O=\pair{0}{0}$ which must be evacuated by the queen and we assume that all robots start from the origin. We use $n$ to denote the number of servants, and use $Q(t)$ and $S_k(t)$, $k = 1,\ \ldots,\ n$, to represent the trajectories of the queen and $k^{th}$ servant respectively. The set of all servant trajectories is represented by $\IS = \{S_k(t);\ k=1,\ldots,\ n\}$. A trajectory will be given as a parametric function of time and, when referring to a robot's trajectory, it will be implied that we mean the path taken by the robot in the case that the exit has not been found.

\subsection{Evacuation algorithms}\label{sec:classes}
A priority evacuation algorithm $\Alg$ is specified by the trajectories of the queen and servants, $\Alg = \{Q(t)\} \cup \IS$. We say that $\Alg$ solves the \pe{n} problem if, in finite time, all points of $\UC$ are visited/discovered by at least one robot. The evacuation time $\IT$ of an algorithm solving the \pe{n} problem is defined to be the worst-case time taken for the queen to reach the exit. As such, the evacuation time will be composed of two parts: the time taken until the exit is discovered plus the time needed for the queen to reach the exit once it has been found.

We will find it useful to define the restricted class of evacuation algorithms $\sclass$ containing all those algorithms in which: a) the queen does not participate in searching for the exit, b) the servants initially move as quickly as they can to the perimeter of $\UC$, c) each servant searches either counter-clockwise or clockwise along the perimeter of $\UC$ at full speed, and, d) each servant stops and is no longer used once it reaches an already discovered point of $\UC$. Algorithms in this class can be defined by the trajectory of the queen $Q(t)$ together with the sets $\Phi = \{\phi_k \in [0,\ 2\pi];\ k=1,\ \ldots,\ n\}$ and $\Sigma = \{\sigma_k=\pm1;\ k=1,\ \ldots,\ n\}$ which respectively specify the angular positions on $\UC$ to which the servants initially move, and the directions in which each servant searches. We will enforce an ordering on the sets $\Phi$ and $\Sigma$ such that for $\phi_k \in \Phi$, $1 \leq k \leq n-1$, we have $\phi_k \leq \phi_{k+1}$. With this notation we can express the trajectory of the $k^{th}$ servant during the time it is searching as 
$S_k(t) = \pair{\coss{\phi_k + \sigma_k (t-1)}}{\sinn{\phi_k+ \sigma_k (t-1)}}$. 

We additionally define the class of algorithms $\symclass \subset \sclass$ containing those algorithms for which we can split the set of servants into two groups $\IS = \IS_+ \cup \IS_-$ where: a) servants in $\IS_+$ follow trajectories which are reflections about the $x$-axis\footnote{The choice of the $x$-axis is arbitrary since we may always rotate $\UC$. What is important is that a diameter of symmetry exists.} of servants in $\IS_-$, and, b) all servants in $\IS_+$ search counter-clockwise \footnote{Again, these choices of search directions are arbitrary since we can reflect $\UC$ about the $y$-axis. What is important is that all servants within a group search in the same direction.}. In the case that $n$ is odd we permit one servant to follow a trajectory that is symmetric about the $x$-axis. For an algorithm in $\symclass$ we may write $\Phi = \Phi_+ \cup \Phi_-$ where $\Phi_+$ (resp. $\Phi_-$) specifies the positions on $\UC$ to which the servants above (resp. below) the $x$-axis initially move. Formally we may write $\Phi_+ = \{\phi_k \in [0, \pi]; k=1,\ \ldots, \left\lceil\frac n2\right\rceil\}$ and $\Phi_- = -\Phi_+$ for even $n$ and $\Phi_- = \{-\phi_k;\ k=2,\ \ldots, \lceil \frac n2 \rceil\}$ for odd $n$. In the class $\symclass$ the directions in which the servants search are always counter-clockwise (resp. clockwise) for robots in $\Phi_+$ (resp. $\Phi_-$) and thus an algorithm $\Alg \in \symclass$ is entirely specified by the set $\{Q(t)\}\cup\Phi_+$.

As a warm-up to the next section, and to demonstrate the intuitive nature of these definitions, consider the following trivial algorithm which achieves an evacuation time of $2 + \frac{2\pi}{n}$: the queen remains at the origin until the exit is found and the servants move directly to equally spaced locations on the perimeter of $\UC$ each searching an arc of length $\frac{2\pi}{n}$ in the counter-clockwise direction. This algorithm can be seen to be in the class $\sclass$ and we can succinctly represent the algorithm as follows

\begin{algorithm}[H] \caption{Trivial Evacuation 1, $\Alg\in\sclass$} \label{alg:trivial1}
    \begin{algorithmic}[1]
        \State $Q(t) = \pair{0}{0}$.
        \State $\Phi = \{\frac{(k-1)}{n}2\pi;\ k=1,\ldots,n\}$
        \State $\Sigma = \{1;\ k=1,\ldots,n\}$
    \end{algorithmic}
\end{algorithm}

Observe that the above algorithm is not in $\symclass$. We can, however, give an equivalent algorithm in $\symclass$ which achieves the same evacuation time. This algorithm is depicted in Figure~\ref{fig:trivial_algos} along with Algorithm~\ref{alg:trivial1} for the case that $n=8$.
\includeFig{width=4in,keepaspectratio}{fig:trivial_algos}{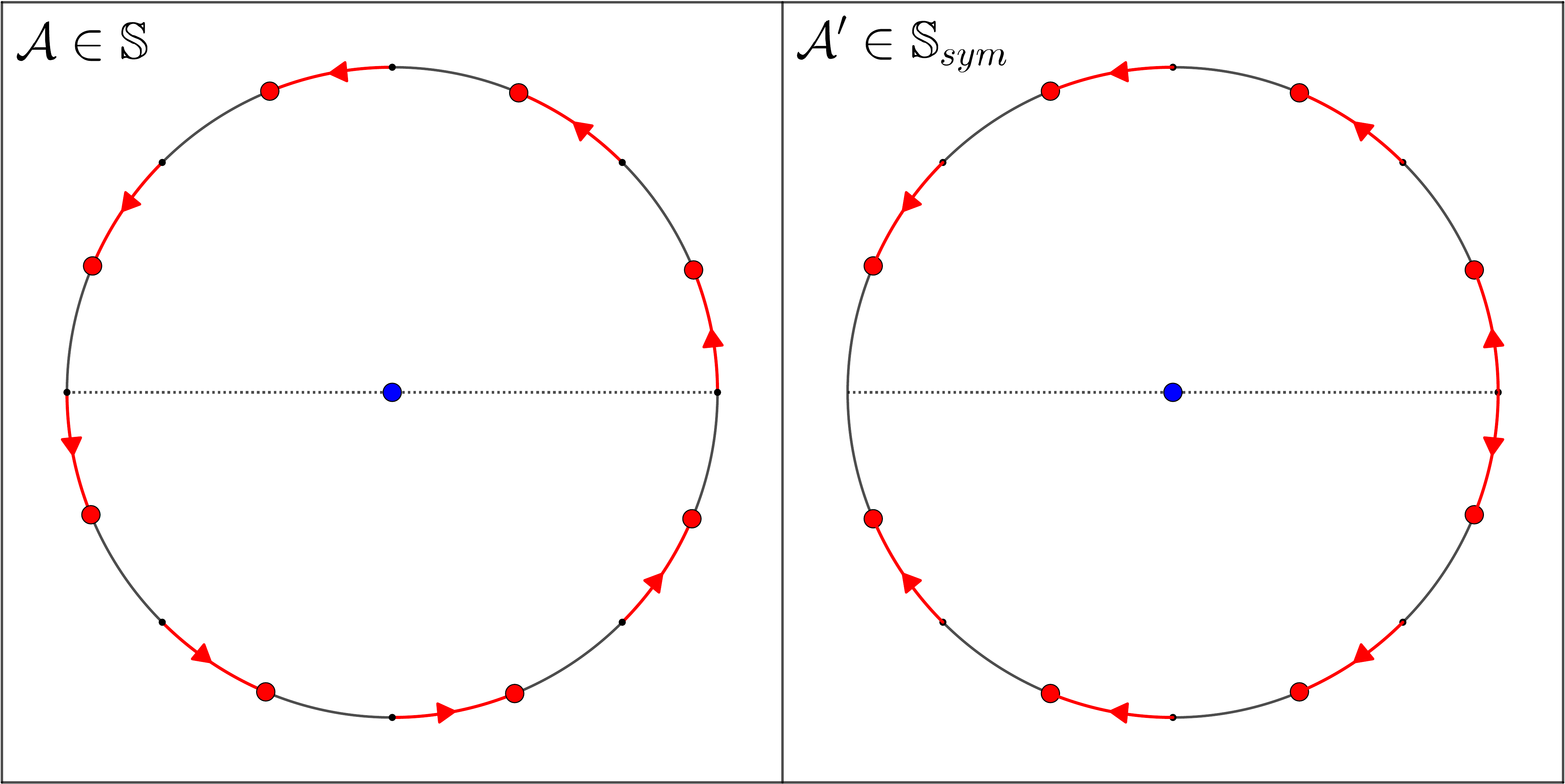}{Depiction of the two trivial algorithms each achieving an evacuation time of $2+\frac{2\pi}{n}$. Both algorithms are in the class $\sclass$ and the algorithm on the right is also in the class $\symclass$. The queen is indicated by the blue point and the servants by the red points. A red arc indicates points that have been discovered.}

\section{Upper Bound}
\label{sec6}
In the previous section we introduced two evacuation algorithms solving \pe{n} with evacuation time $2 + \frac{2\pi}{n}$. We will show that this can be improved:
\begin{theorem}\label{thm:ub}
    There exists an algorithm solving \pe{n} for $n \geq 4$ with an evacuation time at most
    $2 + 4(\sqrt{2}-1) \frac{\pi}{n} \approx 2 + 1.657\frac{\pi}{n}$.
\end{theorem}
We will prove Theorem~\ref{thm:ub} constructively and present an evacuation algorithm in the class $\symclass$ achieving the desired upper bound for $n\geq 4$ servants. For ease of presentation we will assume that $n$ is even. Furthermore, as it will greatly simplify the algebra, we will redefine all times (including the evacuation time) to start from the moment the servants first reach the perimeter. To avoid confusion we will use $\IT_p$ to represent the evacuation time of an algorithm as measured from the moment the servants reach the perimeter. The total evacuation time will thus be $\IT = \IT_p + 1$.

As we will describe an algorithm in the class $\symclass$ we will only need to specify the queen's trajectory $Q(t)$ and the initial angular positions $\Phi_+$ of the servants lying above the $x$-axis. We start by giving the trajectory for the queen which we parametrize using $\alpha > 0$:
\begin{equation}\label{eq:queen}
    Q(t) = \begin{cases}
    \pair{0}{0},& 0 \leq t < \alpha\\
    \pair{\alpha-t}{0},& \alpha \leq t < \alpha+1\\
    \pair{-1}{0},& t \geq \alpha+1,
\end{cases}
\end{equation}
In words, the queen waits at the origin until the time $t=\alpha$ at which moment she begins moving at full speed along the negative $x$-axis stopping when she arrives to the point $\pair{-1}{0}$ at the time $t=\alpha+1$. The crux of the algorithm will be in specifying the set $\Phi_+$. In order to do this we consider the following simple observation:
\begin{observation}\label{obs:disk1}
    If the queen is to achieve an evacuation time of $\IT_p$, then, for all $t < \IT_p$, all of the undiscovered points of $\UC$ must remain inside the disk centered on the queen with radius $\IT_p-t$.
\end{observation}

\includeFig{width=5in,keepaspectratio}{fig:intercepts}{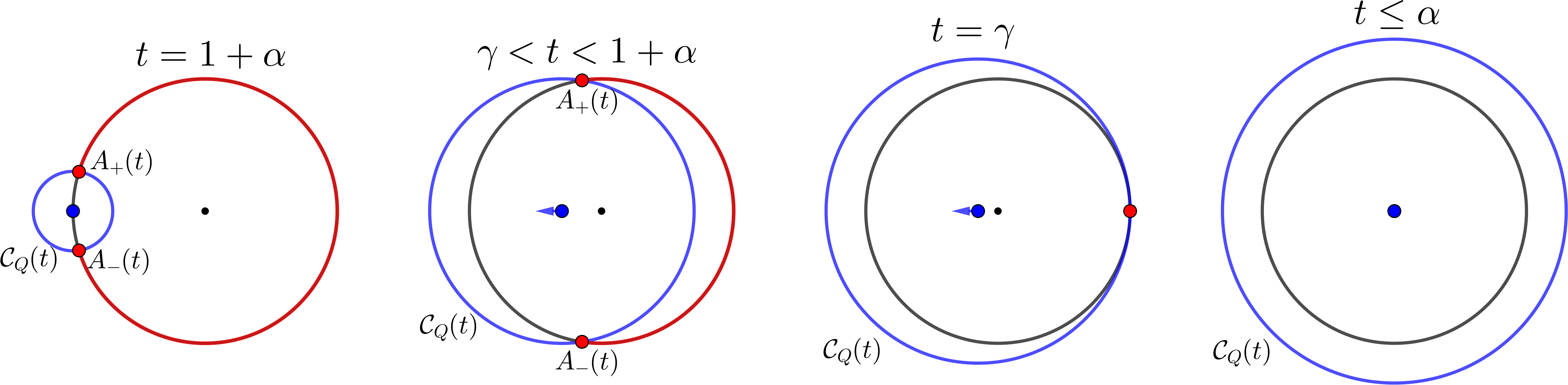}{Illustration of the queen's trajectory $Q(t)$ (blue point) and the motion of the intercepts $A_+(t)$ and $A_-(t)$. The blue circle represents the circle $\IC_Q(t)$ and the black circle represents the circle $\UC$. A red arc indicates those positions of $\UC$ that must be discovered at the indicated time. Time flows from right to left.}

Assume that we have an algorithm with evacuation time $\IT_p$ and define $\IC_Q(t)$ as the circle centered on the queen with radius $\IT_p-t$. Then, in light of Observation~\ref{obs:disk1}, it is not so hard to imagine that the intersection points of the circles $\IC_Q(t)$ and $\UC$ will be of importance. Thus, assume that $\IT_p$ is small enough that at some time $t \geq \alpha$ the circles $\IC_Q(t)$ and $\UC$ intersect. Considering the form of the queen's trajectory, we can conclude that the circles $\UC$ and $\IC_Q(t)$ will first intersect at the time $\gamma = \frac{\IT_p + \alpha - 1}{2}$ at the point $\pair{1}{0}$. For times $t > \gamma$ the circles will intersect at two points $A_\pm$ which are symmetric about the $x$-axis and which move from right to left along the perimeter of $\UC$ (see Figure~\ref{fig:intercepts}). The importance of the points $A_\pm$ is clear when one considers that $A_\pm$ mark the boundary between those points of $\UC$ which must be discovered and those which may yet be undiscovered at the time $t$. Intuitively, we will want to position the servants such that they are searching only when they are to the left of $A_+$ and $A_-$. In particular, a servant will stop searching at precisely the moment the intercept $A_+$ or $A_-$ catches up to it (with a small caveat to be described shortly). This condition will allow us to specify the set $\Phi_+$.
 
At this time we will find it useful to re-express the evacuation time as $\IT_p = 1+\alpha+\rho$ where $\rho$ is a parameter that will ultimately depend on $\alpha$. Intuitively, $\rho$ represents the radius of $\IC_Q(t)$ at the moment the queen reaches the perimeter of $\UC$ and its inclusion will greatly simplify algebra. Note that, with this definition, the circles $\IC_Q(t)$ and $\UC$ will first intersect at the time $\gamma = \alpha + \frac{\rho}{2}$.

As we only need to specify the set $\Phi_+$ we will only consider the intercept $A_+$. The coordinates of $A_+$ for times $\gamma \leq t \leq \alpha+1$ can be determined by simultaneously solving the implicit equations for $\UC$ and $\IC_Q(t)$, i.e. 
    $\UC:\  x^2+y^2=1$
and
    $\IC_Q(t):\  (x-\alpha+t)^2 + y^2 = (1+\alpha+\rho-t)^2$.
We find that $A_+(t) = \pair{x_A(t)}{y_A(t)}$ where
\begin{equation}\label{eq:xA}
    x_A(t) = \frac{\rho(2+\rho)}{2(t-\alpha)} - 1 - \rho
\end{equation}
and
\begin{equation}\label{eq:yA}
    y_A(t) = \frac{\sqrt{\rho (\rho+2)[2(t-\alpha)-\rho][\rho+2-2(t-\alpha)]}}{2(t-\alpha)}
\end{equation}
The angular position of $A_+$ will be represented as $\phi_A$ and is given by:
\begin{equation}\label{eq:phiA}
    \phi_A(t) = \atann{\frac{y_A(t)}{x_A(t)}}.
\end{equation}
We define $\nu_A$ as the speed at which $A_+$ moves along the perimeter of $\UC$. We can determine $\nu_A$ using $\nu_A(t) = \sqrt{\left(\diff{x_A}{t}\right)^2 + \left(\diff{y_A}{t}\right)^2}$ from which we find that:
\begin{equation}\label{eq:nuA}
\nu_A(t) = \frac{1}{t-\alpha}\sqrt{\frac{\rho(\rho+2)}{[\rho+2-2(t-\alpha)][2(t-\alpha)-\rho]}}
\end{equation}

Now consider the form of the function $\nu_A(t)$. For times just after $t=\alpha$ we can see that $A_+$ will move with a speed $\nu_A>>1$ and, as such, no single servant will be able to stay to the left of $A_+$ for long. What is not so obvious from $\eqref{eq:nuA}$ is that $\nu_A$ continuously decreases until some time $\tau$ at which $\nu_A = 1$.\footnote{It is not guaranteed that for all $\rho > 0$ this intercept will reach a speed of one before the queen reaches the perimeter of $\UC$. However, we will choose a $\rho$ such that this does happen.} Furthermore, starting at the time $\tau$ there will be an interval of time during which $\nu_A \leq 1$. Thus, if the intercept reaches a servant at exactly the time $\tau$ that servant does not have to stop searching. We will choose $\rho$ to ensure that the servant $S_{n/2} \in \IS_+$ satisfies exactly this property.

Therefore we can describe the following general overview of our algorithm: the servant $S_1$ begins at $\phi_1=0$ (for even $n$) and searches until the time $t_1$ at which $S_1(t_1) = A_+(t_1)$ or when $t_1 + \phi_1 = \phi_A(t_1)$. The servant $S_2$ will begin its search at the position $\phi_2 = \phi_1 + t_1$ and it will search for a time $t_2$ until $S_2(t_2) = A_+(t_2)$ or until $t_2 + \phi_2 = \phi_A(t_2)$. The servant $S_3$ will begin at the position $\phi_3 = \phi_2+t_2 = \phi_1+t_1+t_2$, and so on. Continuing on like this we can see that the servant $S_k$ will begin its search at the position
$\phi_{k+1} = \phi_{k}+t_{k} = \phi_1+\sum_{i=1}^{k} t_i$
with the $t_k$ satisfying
$t_k = \phi_A(t_k)-\phi_{k}$
or, equivalently,
$\phi_1+\sum_{i=1}^{k} t_i = \phi_A(t_k)$.
We want the servant $S_{n/2}$ to be coincident with the intercept $A_+$ at exactly the time $\tau$ (recall that $\tau$ is the time at which the speed of $A_+$ is $\nu_A = 1$) and thus we will choose $\rho$ to satisfy
$\phi_{n/2} + \tau = \phi_A(\tau)$.
In this case the servant $S_{n/2}$ will search for a total time $\pi-\phi_{n/2}$ after which all of $\UC$ will have been discovered. 

To extend this algorithm to the case that $n$ is odd we will need to split the trajectory of the servant $S_1 \in \IS_+$ between the upper and lower halves of $\UC$. We will therefore start the servant $S_1$ at the position $\phi_1 = \frac{-t_1}{2}$. All of the other relevant equations remain unchanged.

We provide links (\cite{anim4} and \cite{anim8}) to short animations of the algorithm for $n=4,\ 8$. In these animations the queen is represented by the blue point, the servants by red points, and the intercepts $A_\pm$ by green points. A plot of the evacuation time as a function of the time at which the servants find the exit is also shown. Note that the servants stop searching at the exact moment the intercept reaches them (except for the two servants furthest to the left) and at these moments the evacuation time is maximized. The two servants that are last active will be coincident with the intercepts at the moment these intercepts reach a speed of one, and, again, at this moment the evacuation time is maximized. In total there will be $n$ different locations for the exit (counting the top and bottom of $\UC$) which will maximize the evacuation time. A keen eye will note that the queen reaches the perimeter of $\UC$ before the servants have finished searching the perimeter and this would appear to hint that Algorithm~\ref{alg:ns} can be improved. We will argue in Section~\ref{secconclusion} that this is not the case.

Figure~\ref{fig:n_8} illustrates an example configuration for the described algorithm when $n=8$. The algorithm is formally presented in Algorithm~\ref{alg:ns} where we have left $\alpha$ as a parameter. We claim that Algorithm~\ref{alg:ns} will always do better than the bound of Theorem~\ref{thm:ub} when the evacuation time is minimized over $\alpha$. We will now prove this claim.

\begin{algorithm}[H] \caption{IntersectChase($\alpha$), $\Alg_\alpha \in \symclass$} \label{alg:ns}
    \begin{algorithmic}[1]
        \State 
        \begin{equation*}
            Q(t) = \begin{cases}
            \pair{0}{0},& 0 \leq t < \alpha\\
            \pair{\alpha-t}{0},& \alpha \leq t < \alpha+1\\
            \pair{-1}{0},& t \geq \alpha+1,
        \end{cases}
        \end{equation*}
        
        \State $\Phi_+ = \{\phi_k;\ k=1,\ \ldots,\ \lceil\frac{n}{2}\rceil\}$, where:
        \begin{equation*}\label{eq:phi1}
            \phi_1 = \begin{cases}
                0, & n \mbox{ even} \\
                -\frac{t_1}{2}, & n \mbox{ odd}
            \end{cases},\quad\quad \phi_{k} =  \phi_1 + \sum_{i=1}^{k-1} t_i, \quad\quad \phi_{n/2} + \tau = \phi_A(\tau)
        \end{equation*}
        and,
        \begin{equation*}\label{eq:tk}
            \phi_1 + \sum_{i=1}^{k} t_i = \phi_A(t_k), \quad\quad \nu_A(\tau) = 1
        \end{equation*}
    \end{algorithmic}
\end{algorithm}

\includeFig{width=3.5in,keepaspectratio}{fig:n_8}{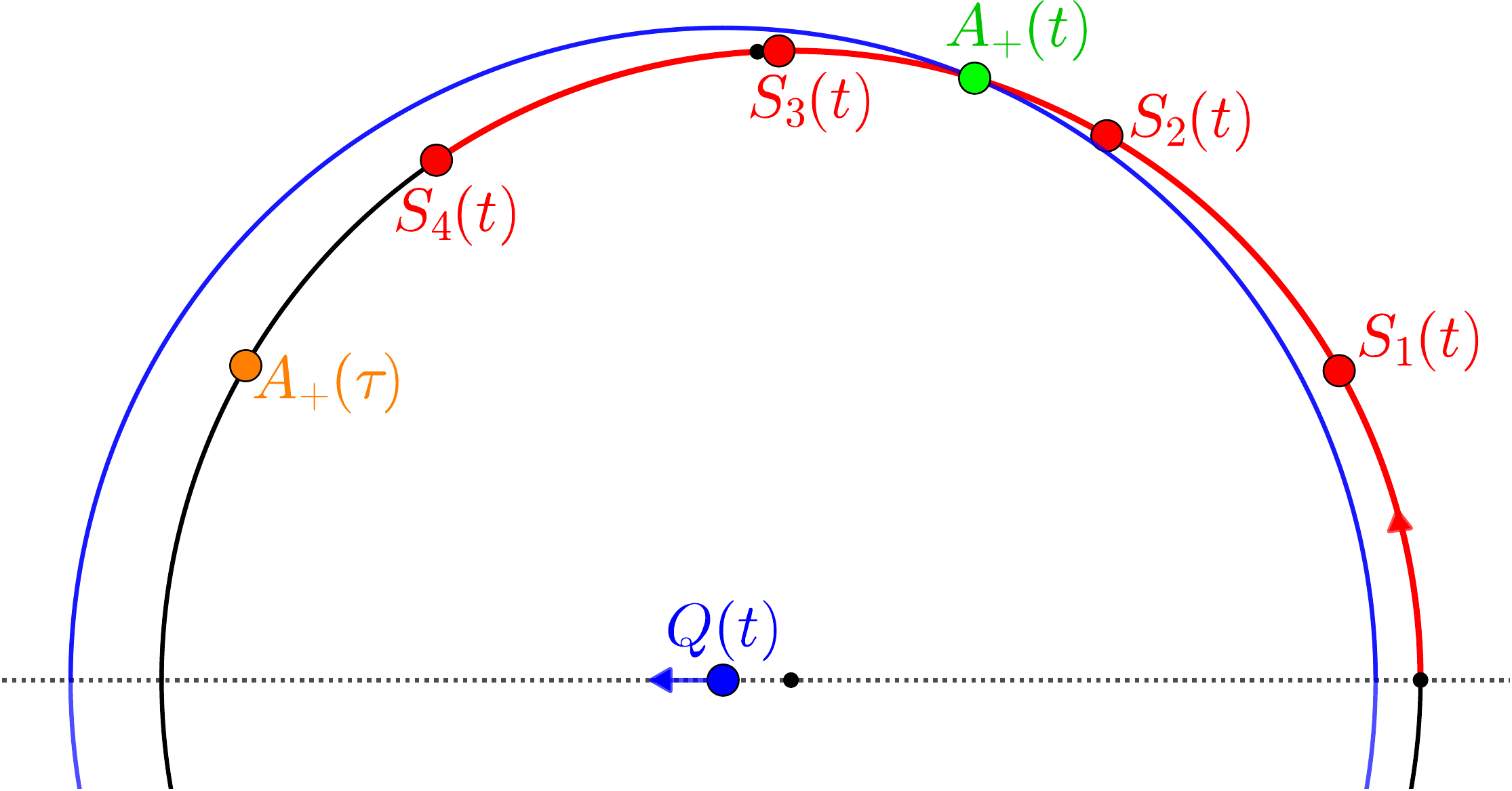}{Example configuration of Algorithm~\ref{alg:ns} when $n=8$. The configuration is only shown for the 4 servants on the upper half of the circle $\UC$. In this diagram all servants move counter-clockwise. The servants $S_1$ and $S_2$ have already finished their search and are located at the starting positions of the respective servants $S_2$ and $S_3$. The servant $S_3$ is just about to finish its search. The point $A_+(\tau)$ marks the location where the intercept $A_+(t)$ slows to a speed of 1. The servant $S_4$ will reach the point $A_+(\tau)$ at the exact moment the intercept does.}


\begin{proof} (Theorem~\ref{thm:ub})
    To simplify the algebra we will assume that $n$ is even. Algorithm~\ref{alg:ns} specifies that we choose the $t_k$ in order to satisfy $\sum_{i=1}^{k} t_i = \phi_A(t_k)$ where $\phi_A(t)$ is defined in \eqref{eq:phiA}. We note that each servant will be able to search for at least a time $\gamma$ since this marks the first time at which $\IC_Q(t)$ and $\UC$ intersect. This motivates us to define the primed time coordinate  $t' = t - \gamma$. In this primed coordinate the defining relation for the $t'_k$ is
        $\sum_{i=1}^{k} t'_k = \phi_A(t'_k) - k\gamma$
    (where we assume that $\phi_A$ is properly redefined for the primed time coordinate). We are interested in an asymptotic limit and thus we make the following claim:
    \begin{claims}\label{clm:sum}
        When we take the limit in large $n$, the sum $\sum_{i=1}^{k} t'_i$ becomes a definite integral
        $\lim_{n \rightarrow \infty} \sum_{i=1}^{k} t'_i = \int_{0}^{\kappa} t'(u) du$
        where $\frac{\kappa}{n}$ is to be interpreted as the fractional servant number and $u$ is a dummy integration variable.
    \end{claims}
        \begin{proof} 
        Consider the sum $\sum_{i=1}^{k} t'_i$. Define $f_k = \frac{k}{n}$ as the fractional servant number and redefine $t'_k$ to be a function of $f_k$, i.e., $t'_k = t'(f_k)$. Also define $\Delta f = \Delta f_k = \frac{1}{n}$. Finally, defining $t'(0) = 0$ allows us to rewrite the sum as $\sum_{i=1}^{k} t'_i = n \sum_{f_i=0}^{f_k} t'(f_i) \Delta f_i$.
    Now, for a given constant fraction $f_k \leq \frac{1}{2}$ the bounds of the sum $\sum_{f_i=0}^{f_k} t'(f_i) \Delta f_i$ are constant. Furthermore, in the limit $n \rightarrow \infty$ the interval $\Delta f_i \rightarrow 0$. Thus, the limit,
    $\lim_{n \rightarrow \infty} \sum_{f_i=0}^{f_k} t'(f_i) \Delta f_i$
    is simply the definition of the Riemann integral of $t'(f)$ over the domain $f \in [0, f_k]$, i.e., 
    $\lim_{n \rightarrow \infty} \sum_{f_i=0}^{f_k} t'(f_i) \Delta f_i = \int_{0}^{f_k} t'(f) df$.
    By defining $\kappa = n f_k$ we get $\int_{0}^{f_k} t'(f) df = \frac{1}{n} \int_{0}^{\kappa} t'(u) du$ and this leads us to our desired result
    $\lim_{n \rightarrow \infty} \sum_{i=1}^{k} t'_i = \int_{0}^{\kappa} t'(u) du$.
    \qed \end{proof}
    Due to the Claim~\ref{clm:sum}, the asymptotic defining relation for $t'(\kappa)$ becomes an integral equation
    $\int_{0}^{\kappa} t'(u) du = \phi_A(t'(\kappa)) - \kappa \gamma$.
    Using the fundamental theorem of calculus we can rewrite this as a differential equation:
    $t'(\kappa) = \frac{d}{d\kappa}(\phi_A(t'(\kappa)) - \kappa\gamma) = \frac{d \phi_A(t'(\kappa))}{d\kappa} - \gamma$.
    Applying the chain rule we find that 
    $\frac{d \phi_A(t'(\kappa))}{d\kappa} = \frac{d \phi_A(t'(\kappa))}{dt'} \cdot \frac{d t'(\kappa)}{d\kappa}$.
    Observe that $\frac{d \phi_A(t'(\kappa))}{dt'}$ is simply the speed of the intercept $A_+$ and we can therefore write the differential equation for $t'(\kappa)$ as
    $\frac{dt'}{d\kappa} = \frac{t'+\gamma}{\nu_A(t'(\kappa))}$.
    This ordinary differential equation can easily be solved for $\kappa$ in terms of $t'$ by separation of variables. We find that 
    $\kappa(t') = \int_{0}^{t'} \frac{\nu_A(u)}{u+\gamma} du$.
    The equation for the speed $\nu_A$ is given in \eqref{eq:nuA}, which, in the primed time coordinate takes the form
    $\nu_A(t') = \frac{1}{(2t'+\rho)}\sqrt{\frac{\rho (\rho + 2)}{t'(1-1')}}$.
    Substituting this into the expression for $\kappa(t')$ yields 
    $\kappa(t') = \int_{0}^{t'} \frac{1}{(u+\gamma)(2u+\rho)}\sqrt{\frac{\rho (\rho + 2)}{u(1-u)}} du$.
    This integral has the closed form solution
    \begin{align*}
        \kappa(t') &= \frac{1}{\alpha} \left[ 2 \tan^{-1}\left(\frac{t'(2t'+\rho)}{\rho} \nu_A(t') \right) \right. \\
          & \left. - \sqrt{\frac{\rho(\rho+2)}{\gamma(\gamma+1)}} \tan^{-1}\left(t'(2t'+\rho) \sqrt{\frac{1+\gamma}{\gamma \rho(\rho+2)}}\nu_A(t') \right) \right].
    \end{align*}
    We require that the servant $S_{n/2}$ be coincident with the intercept $A_+$ at the time $\tau' = \tau-\gamma$ and this implies that we need $\kappa(\tau') = \frac{n}{2}$ or
    \begin{align*}
        \frac{n}{2} &= \frac{1}{\alpha} \left[ 2 \tan^{-1}\left(\frac{\tau'(2\tau'+\rho)}{\rho} \nu_A(\tau') \right)
        \right. \\ &- \left. \sqrt{\frac{\rho(\rho+2)}{\gamma(\gamma+1)}} \tan^{-1}\left(\frac{\tau'(2\tau'+\rho)}{\rho} \sqrt{\frac{\rho(\gamma+1)}{\gamma(\rho+2)}}\nu_A(\tau') \right) \right].
    \end{align*}
    If we set $\alpha = a\frac{\pi}{n}$ and note that, by definition, $\nu_A(\tau') = 1$, we can simplify the above to obtain 
    \begin{align*}
        \frac{\pi}{2} = \frac{1}{a} \left[ 2 \tan^{-1}\left(\frac{\tau'(2\tau'+\rho)}{\rho} \right)
        - \sqrt{\frac{\rho(\rho+2)}{\gamma(\gamma+1)}} \tan^{-1}\left(\frac{\tau'(2\tau'+\rho)}{\rho} \sqrt{\frac{\rho(\gamma+1)}{\gamma(\rho+2)}} \right)\right].
    \end{align*}
    Define $D(a, \rho)$ as the quantity
    \begin{align*}
        D(a, \rho) &= \frac{\pi}{2} - \frac{1}{a} \left[ 2 \tan^{-1}\left(\frac{\tau'(2\tau'+\rho)}{\rho} \right)
       \right. \\ &- \left. \sqrt{\frac{\rho(\rho+2)}{\gamma(\gamma+1)}} \tan^{-1}\left(\frac{\tau'(2\tau'+\rho)}{\rho} \sqrt{\frac{\rho(\gamma+1)}{\gamma(\rho+2)}} \right)\right]
    \end{align*}
    which we want to be zero. We now make the following claim:
    
    \begin{claims}\label{clm:tau}
        The asymptotic behaviour of $\tau$ is $\BO{\rho^{1/3}}$.
    \end{claims}

    \begin{proof} 
    We want to show that $\tau$ has asymptotic behaviour $\BO{\rho^{1/3}}$. One way that we can do this is by formally computing a Puiseux series of $\tau$ from which one would find that the first few terms in the expansion are
    $\tau = \left(\frac{\rho}{2}\right)^{1/3} + \frac{1}{3}\left(\frac{\rho}{2}\right)^{2/3} + \BO{\rho}$.
    Alternatively, we note that $\tau'$ solves the equation $\nu_A(\tau') = 1$, i.e.
    $\frac{1}{2t'+\rho}\sqrt{\frac{\rho (\rho + 2)}{t'(1-t')}} = 1$.
    Expanding the above we arrive at the quartic equation
    $t''^4 - (\rho+1)t''^3 + \frac{\rho(\rho+2)}{4}(t''^2 + 1) = 0$
    where $t'' = t'+\frac \rho 2 = t - \alpha$. As $n \rightarrow \infty$ we have $1+t''^2 \rightarrow 1$, $2+\rho \rightarrow 2$ and $t''^4 \rightarrow 0$. As $n$ gets large $\tau$ approaches the solution to
    $-t''^3 + \frac \rho 2 = 0$
    which also demonstrates that $\tau = \BO{\rho^{1/3}}$. 
    \qed \end{proof}

Using Claim~\ref{clm:tau}, we have that
    $\lim_{n\rightarrow \infty} \frac{\tau'(2\tau'+\rho)}{\rho} = \lim_{n\rightarrow \infty} \BO{\rho^{-1/3}} = \infty$
    and thus
    \begin{align*}
        \frac \pi 2 = \lim_{n\rightarrow \infty} \tan^{-1}\left(\frac{\tau'(2\tau'+\rho)}{\rho} \right)
         = \lim_{n\rightarrow \infty} \tan^{-1}\left(\frac{\tau'(2\tau'+\rho)}{\rho} \sqrt{\frac{\rho(\gamma+1)}{\gamma(\rho+2)}} \right).
    \end{align*}
    We can therefore write
    $\lim_{n\rightarrow \infty} D(a,\rho) = \frac{\pi}{a} \left(1 - \frac{a}{2} - \sqrt{\frac{\rho}{2\gamma}}\right)$.
    Now set $\rho = q\frac{\pi}{n}$ such that $\gamma = \alpha + \frac \rho 2 = \frac{\pi}{n}(a+\frac{q}{2})$. Using this notation we have
    $\lim_{n\rightarrow \infty} D(a,\rho) = \frac{\pi}{a} \left(1 - \frac{a}{2} - \sqrt{\frac{q}{2a+q}}\right)$.
    We want this limit to equal zero which implies that we need $1 - \frac{a}{2} - \sqrt{\frac{q}{2a+q}} = 0$ or $q = \frac{2(2-a)^2}{(4-a)}$.
    
    Now, to optimize the algorithm we need to minimize the evacuation time $\IT_p$. Since $\IT_p$ increases with $a$ we equivalently need to minimize $a+q = \frac{a^2-4a+8}{4-a}$. Taking the derivative of this with respect to $a$ and setting the result equal to zero gives us the optimal value of $a$ and $q$ to be $a = 2(2-\sqrt{2})$ and $q = 2(3\sqrt{2}-4)$. The asymptotic cost of the algorithm is therefore
    $\IT_p = 1+\alpha+\rho = 1 + 4(\sqrt{2}-1) \frac{\pi}{n}$.
    The overall evacuation time is then $\IT = 1+\IT_p$ which is the bound given in Theorem~\ref{thm:ub}. 

    We note that, in the case that $n$ is odd, the results of the proof will not change due to the fact that, as $n\rightarrow \infty$, we have $\phi_1 = -\frac{t_1}{2} \rightarrow 0$.  
    \qed \end{proof}

\section{Lower Bound}
\label{sec7}

In this section we develop a lower bound on the evacuation time of the queen. We first note that we can derive a naive lower bound of $2+\frac{\pi}{n+1}$ since each robot can travel with a maximum speed of one and we have $n+1$ robots in total. We will show that this can be improved:

\begin{theorem}\label{thm:lb}
    In any algorithm with $n\geq 4$ the queen cannot be evacuated in time less than
    $1 + \frac{2}{n}\cos^{-1}\left(\frac{-2}{n}\right) + \sqrt{1 - \frac{4}{n^2}}$.
    In the limit of large $n$ this bound approaches
    $2+\frac{\pi}{n}+\frac{2}{n^2}$.
\end{theorem}

The outline of the proof is as follows: we first demonstrate that the lower bound holds for any algorithm in which the queen does not participate in searching for the exit before some critical time. We will then show that the queen is not able to participate in the search for the exit before this critical time. 
We begin with a lemma first given in \cite{CGGKMP} which is reproduced here for convenience:

\begin{lemma}\label{lm:chord}
    Consider a perimeter of a disk whose subset of total length $u + \epsilon > 0$ has not been explored for some $\epsilon > 0$ and $\pi \geq u > 0$. Then there exist two unexplored boundary points between which the distance along the perimeter is at least $u$.
\end{lemma}

In the next two lemmas we demonstrate that the lower bound holds if the queen does not participate in the search.

\begin{lemma}\label{lm:lb1}
    For $n \geq 2$, any $x$ satisfying $\frac{\pi}{n} \leq x < \frac{2\pi}{n}$, and any evacuation algorithm in which the queen does not participate in searching for the exit before the time $1+x$, it takes time at least $1+x+\sinn{\frac{nx}{2}}$ to evacuate the queen.
\end{lemma}

\begin{proof}
    Consider an algorithm $\Alg$ with evacuation time $\IT$ and with $n$ servants. Then, at the time $t = 1+x$, the total length of perimeter that the robots have explored is at most $nx \geq \pi$ (since each robot may search at a maximum speed of one, the queen does not search by assumption, and the servants need at least a unit of time to reach the perimeter). Thus, by Lemma~\ref{lm:chord}, there exists two unexplored points on the perimeter of $\UC$ whose distance along the perimeter is at least $2\pi-nx-\epsilon$ for any $\eps>0$. The chord connecting these points has length at least $2\sinn{\pi - \frac{nx}{2}-\frac{\eps}{2}}$ and an adversary may place the exit at either endpoint of this chord. The queen will therefore take at least $\sinn{\pi - \frac{nx}{2}-\frac{\eps}{2}}$ more time to evacuate and the total evacuation time will be at least $1+x+\sinn{\pi - \frac{nx}{2}-\frac{\eps}{2}}$. As this is true for any $\eps > 0$ taking the limit $\eps \rightarrow 0$ we obtain $\IT \geq 1+x+\sinn{\pi-\frac{nx}{2}} = 1+x+\sinn{\frac{nx}{2}}$.
\qed \end{proof}

\begin{lemma}\label{lm:lb2}
    For any $n \geq 2$ and any evacuation algorithm in which the queen does not participate in searching for the exit before the time $t = 1+\frac{2}{n}\acoss{\frac{-2}{n}}$ it takes time at least $1+\frac{2}{n}\acoss{\frac{-2}{n}}+\sqrt{1-\frac{4}{n^2}}$ to evacuate the queen.
\end{lemma}

\begin{proof} 
    Set $f(x) = 1+x+\sinn{\frac{nx}{2}}$. The maximum value of $f(x)$ occurs when $\diff{f}{x}=0$ or when $x = \frac{2}{n}\acoss{\frac{-2}{n}}$. Since $\frac{\pi}{n} \leq \frac{2}{n}\acoss{\frac{-2}{n}} < \frac{2\pi}{n}$ we can invoke Lemma~\ref{lm:lb1} to get a lower bound on the evacuation time of $\IT \geq 1 + \frac{2}{n}\acoss{\frac{-2}{n}} + \sinn{\acoss{\frac{-2}{n}}} = 1+\frac{2}{n}\acoss{\frac{-2}{n}}+\sqrt{1-\frac{4}{n^2}}$, provided that the queen does not search before the time $t = 1+\frac{2}{n}\acoss{\frac{-2}{n}}$.
\qed \end{proof}

We will now demonstrate that the queen is not able to search before the time $1+\frac{2}{n}\acoss{\frac{-2}{n}}$. This will be the goal of the next four lemmas and the following simple observation

\begin{observation}\label{obs:disk2}
    If the queen is to achieve an evacuation time of $\IT$, then, for any time $t \leq \IT$, she must remain in the region of intersection of all disks centered on the undiscovered points of $\UC$ with radii $\IT-t$.
\end{observation}

\begin{lemma}\label{lm:AB}
    Consider any two points $A$ and $B$ on the unit circle connected by a chord of length $\delta$. Define the circles $\IC_A$ and $\IC_B$ as the circles centered on $A$ and $B$ with radii $r$. Then, if $r > \frac \delta 2$, the circles intersect at two points $C$ and $D$ at distances $\sqrt{r^2 - \frac{1}{4}\delta^2} \pm \sqrt{1 - \frac{1}{4}\delta^2}$ from the origin.
\end{lemma}
\begin{proof}    
    Assume that $r > \frac 12 \delta$. Set $C$ and $D$ as the intersection points of $\IC_A$ and $\IC_B$, and set $E$ as the midpoint of $A$ and $B$. Refer to Figure~\ref{fig:AB} for a setup of the proof.
  
    \includeFig{width=3in,keepaspectratio}{fig:AB}{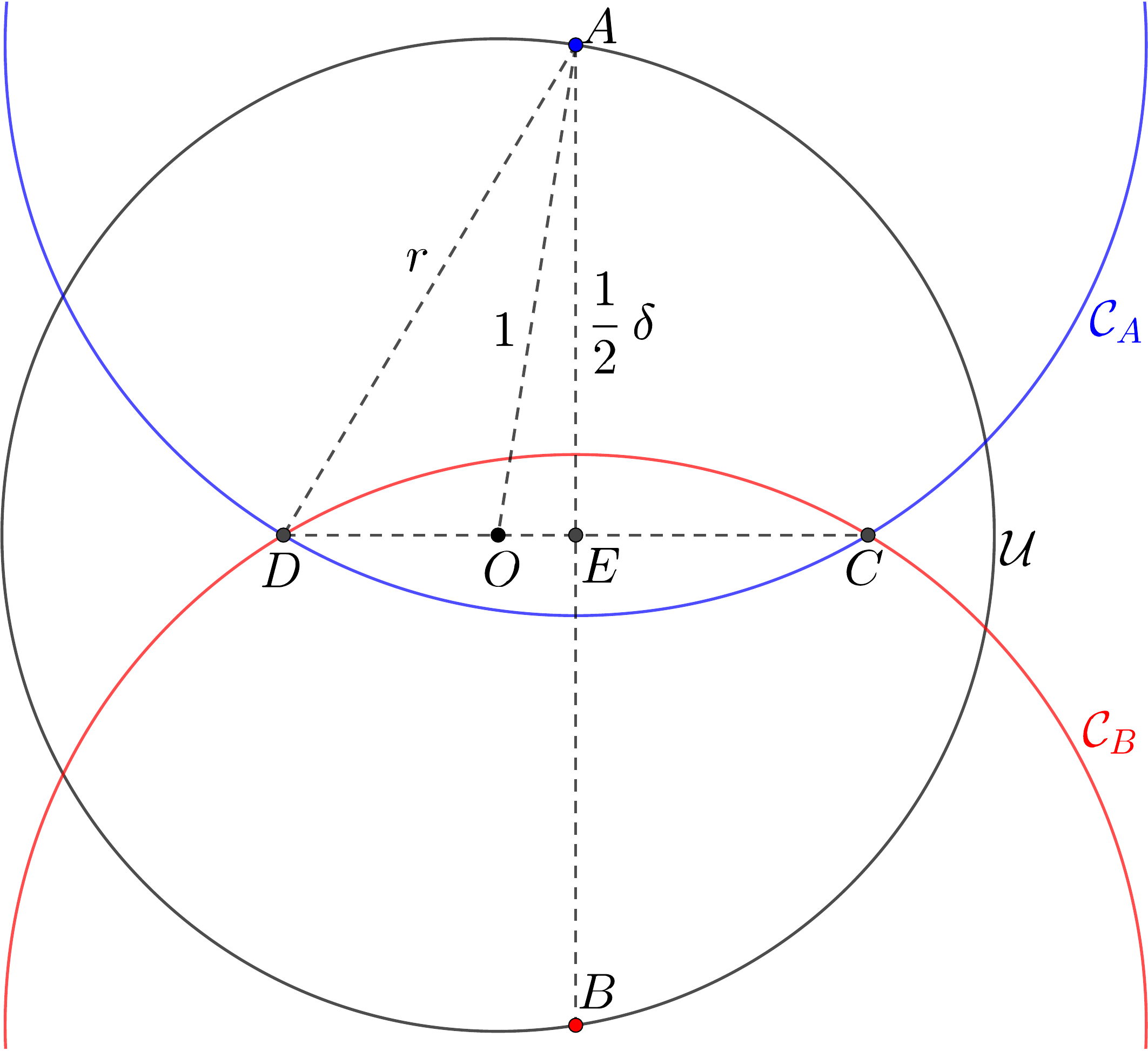}{Setup for the proof of Lemma~\ref{lm:AB}.}

    Since $\IC_A$ and $\IC_B$ have the same radius the points $C$ and $D$ are each separately equidistant to $A$ and $B$ and they will therefore lie on the perpendicular bisector of $A$ and $B$. Since $A$ and $B$ lie on the unit circle this bisector will pass through the origin. Referring to Figure~\ref{fig:AB} it is therefore clear that\footnote{$\abs{AB}$ represents the Euclidean distance between two points $A$ and $B$.} $|EC| = |ED| = \sqrt{r^2 - \frac 14 \delta^2}$ and $|OE| = \sqrt{1 - \frac 14 \delta^2}$. Since $|OC| = |OE| + |EC|$ and $|OD| = |ED| - |OE|$ we get the desired result.
\qed \end{proof}
\begin{lemma}\label{lm:fpm}
    For a given $r>0$ define the functions
    $f_{\pm}(x) = \frac{1}{2}\sqrt{4r^2 - x^2} \pm \frac 12 \sqrt{4 - x^2}$.
    Then, for $0 \leq x \leq \min\{2,\ 2r\}$ $f_+$ is a decreasing function of $x$ and $f_-$ is an increasing function of $x$ if $r > 1$ otherwise it is decreasing.
\end{lemma}

\begin{proof}
    Assume that $x$ is in the interval $0 \leq x \leq \min\{2,\ 2r\}$. The rate of change of $f_{\pm}$ with $x$ is
    $\frac{df_\pm}{dt} = \frac{-x}{2}\left(\frac{1}{\sqrt{4r^2-x^2}} \pm \frac{1}{\sqrt{4-x^2}}\right)$.
    For $x$ in the given interval it is clear from this expression that $\frac{df_+}{dt} < 0$. On the other hand, $\frac{df_-}{dt}<0$ when
    $\frac{1}{\sqrt{4r^2-x^2}} - \frac{1}{\sqrt{4-x^2}} > 0$.
    It is not hard to see that this occurs when $r < 1$.
\qed \end{proof}

\begin{lemma}\label{lm:distance}
    Consider any $r >0$ and assume that the unexplored subset of $\UC$ has total length $\phi$. Define $\ID_P$ as the disk centered on an undiscovered point $P \in \UC$ with radius $r$ and define $\IG$ as the region of intersection of all such disks. Then, if $r\geq\sinn{\frac{\phi}{2}}$, $\IG$ is completely contained inside of a disk centered on the origin with radius
    $R = \sqrt{r^2 - \sin^2\left(\frac{\phi}{2}\right)} + \coss{\frac \phi 2}$.
    If $r < \sin{\frac{\phi}{2}}$ then $\IG = \emptyset$.
\end{lemma}

\begin{proof}
    By Lemma~\ref{lm:chord} there exists two undiscovered points $A,\ B \in \UC$ such that the length between them along the perimeter of $\UC$ is at least $\phi$. Take these two points and set $\delta \geq 2\sinn{\frac{\phi}{2}}$ as the length of the chord connecting them. Since $A$ and $B$ are unexplored there must exist a corresponding point $G \in \IG$ such that $G$ is at most a distance $r$ from both $A$ and $B$. Since the chord connecting $A$ and $B$ has length $\delta$, $G$ is at least a distance $\frac{\delta}{2} \geq \sinn{\frac{\phi}{2}}$ from one of $A$ and $B$. We can therefore conclude that in order for $G$ to exist we must have $r \geq \sinn{\frac{\phi}{2}}$. This proves the second part of the lemma.
 
    Take $r \geq \sinn{\frac{\phi}{2}}$ and define the functions $f_{\pm}(x)$ as in Lemma~\ref{lm:fpm}. Observe that
    \[R = \begin{cases}
            f_+\left( 2\sinn{\frac \phi 2} \right),& 0 \leq \phi \leq \pi\\
            f_-\left( 2\sinn{\frac \phi 2} \right),& \pi < \phi \leq 2\pi.
        \end{cases}
    \]
    We consider the cases $\phi \leq \pi$ and $\phi > \pi$ separately.\\

    \noindent \textbf{Case 1:} Take $\phi \leq \pi$ and assume that the lemma is false. In this case there must be a point on the boundary of $\IG$ that is at a distance $d_G > f_{+}(2 \sinn{\frac \phi 2})$ from the origin. Let $G$ be such a point. Since $G$ is on the boundary of $\IG$ there are two undiscovered points $A,\ B \in \UC$ that are at a distance $r$ from $G$. The point $G$ is therefore a point of intersection of two circles $\IC_A$ and $\IC_B$ centered on $A$ and $B$ with radii $r$. If $\delta$ is the length of the chord connecting $A$ and $B$ then by Lemma~\ref{lm:AB} the point $G$ lies a distance $f_{+}(\delta)$ or $f_-(\delta)$ from the origin. Since $f_+(x) \geq f_-(x)$ then $G$ must in fact lie a distance $f_{+}(\delta)$ from the origin. By assumption $d_G > f_{+}\left(2 \sinn{\frac \phi 2}\right)$ and thus $f_{+}(\delta) > f_{+}\left(2 \sinn{\frac \phi 2}\right)$. Since $f_{+}$ is a decreasing function we can further say that $\delta < 2 \sinn{\frac \phi 2}$ and we can thus conclude that there cannot exist two undiscovered points $A,\ B \in \UC$ a distance $\delta \geq  2 \sinn{\frac \phi 2}$ from each other. However, Lemma~\ref{lm:chord} states there must exist two undiscovered points in $\UC$ such that the chord joining them has length at least $2 \sinn{\frac \phi 2}$. We have arrived at a contradiction and must therefore accept that the lemma is valid in the case that $\phi \leq \pi$.\\

    \noindent \textbf{Case 2:} Take $\phi > \pi$ and assume that the lemma is false. In this case there must be a point $G$ on the boundary of $\IG$ that is a distance $d_G > f_-\left(2 \sinn{\frac \phi 2}\right)$ from the origin. Since $G$ is on the boundary of $\IG$ there are two undiscovered points $A,\ B \in \UC$ that are at a distance $r$ from $G$, and, as before, if $\delta$ is the length of the chord connecting $A$ and $B$, the point $G$ must lie a distance $f_{+}(\delta)$ or $f_-(\delta)$ from the origin. Assume first that $G$ is at a distance $f_-(\delta)$. Then by assumption that $d_G > f_{-}\left(2 \sinn{\frac \phi 2}\right)$ we get $f_{-}(\delta) > f_{-}\left(2 \sinn{\frac \phi 2}\right)$. By Lemma~\ref{lm:fpm} and assumption that $r>1$ we can conclude that $f_-$ increases with $x$. We can thus say that $\delta > 2 \sinn{\frac \phi 2}$. Furthermore, since $G \in \IG$, there cannot be any undiscovered points in $\UC$ that are a distance $r$ from $G$. This implies that the smaller arc connecting $A$ and $B$ must be discovered and as such $\phi \leq 2\pi - 2\sin^{-1}\left(\frac \delta 2 \right)$ or $\delta \leq 2\sinn{\pi-\frac \phi 2} = 2\sinn{\frac \phi 2}$. This, however, contradicts with our previously derived condition that $\delta > 2\sinn{\frac \phi 2}$. Now assume that $G$ is at a distance $f_+(\delta)$ from the origin. In a similar manner to the previous case we can conclude that the larger arc connecting $A$ and $B$ must be discovered. This arc, however, has a total length at least $\pi$ and thus $\phi \leq \pi$. This also leads to a contradiction since we have assumed that $\phi > \pi$. We must therefore conclude that the lemma is also valid in the case that $\phi > \pi$.
\qed \end{proof}

\begin{lemma}\label{lm:queen_loc}
    Consider an algorithm with evacuation time $\IT<3$. Then if the queen is able to search the perimeter of $\UC$ we must have 
    $$R(t) = \sqrt{(\IT-t)^2 - \sin^2\left(\frac{n(t-1)}{2}\right)} - \coss{\frac {n(t-1)} 2} > 1.$$
\end{lemma}
\begin{proof}
    By Observation~\ref{obs:disk2} and Lemma~\ref{lm:distance} we can immediately conclude that the queen must be within a distance $R(t) = \sqrt{(\IT-t)^2 - \sin^2\left(\frac{\phi(t)}{2}\right)} + \coss{\frac {\phi(t)} 2}$ of the origin at any time $t$. Consider the time $t=1$. At this time the robots have not been able to search any of the perimeter of $\UC$ and thus $\phi(1) = 2\pi$. At this time the queen can be a distance at most $R(1) = \IT-2$ from the origin and, since we have assumed that $\IT < \IT_0 < 3$, we have $R(1) < 1$. Therefore the queen cannot be on the perimeter of $\UC$ at the time $t=1$.
    
    Now, for a given fixed time $t$ consider how $R$ changes with $\phi(t)$. We find that
    \[\diff{R}{\phi} = -\frac{1}{2}\sinn{\frac{\phi}{2}}\left[1+\frac{\coss{\frac{\phi}{2}}}{\sqrt{(\IT-t)^2 - \sin^2\left(\frac{\phi(t)}{2}\right)}}\right].\]
    First consider the case that $\pi < \phi \leq 2\pi$. We want to determine when $R$ is increasing with $\phi$. Since $\coss{\frac{\phi}{2}}<0$ for  $\pi < \phi \leq 2\pi$, $R$ can only increase with $\phi$ if $\sqrt{(\IT-t)^2 - \sin^2\left(\frac{\phi(t)}{2}\right)} < \cos{\frac{\phi}{2}}$ or if $\IT-t < 1$. Now assume that at the time $t$ the robots have searched the perimeter at a rate of $\mu$ such that $\phi = 2\pi - \mu t$. Since we are considering the case that $\phi > \pi$ we need $\pi > \mu t$. To have $\phi$ increasing we needed $\IT-t < 1$ or $t > \IT-1$ and thus we must have $\pi > \mu (\IT-1)$ or $\mu < \frac \pi{\IT-1}$. A trivial lower bound on $\IT$ is 2 and thus $\mu < \pi$. We claim that this case can be ignored since the robots will need to search at a much higher rate if they are to achieve the lower bound of Theorem~\ref{thm:lb}. 
    
    In the second case both $\coss{\frac \phi 2} \geq 0 $ and $\sinn{\frac \phi 2}\geq0$ for $0 \leq \phi \leq \pi$. Thus, it is not possible that $R$ increases with $\phi$. 
    
    We can conclude from the above analysis that $R$ decreases with $\phi$ in all reasonable cases and thus we maximize $R$ when $\phi(t)$ is minimized. Since the queen cannot be on the perimeter of $\UC$ at the beginning of the algorithm the robots can search at most at a rate $n$ and therefore $\phi(t)$ is minimized when $\phi(t) = 2\pi-nt$. Thus, up until the time the queen reaches the perimeter of $\UC$, the robots must be located within a distance $R(t) = \sqrt{(\IT-t)^2 - \sin^2\left(\frac{nt}{2}\right)} - \coss{\frac {nt} 2}$ of the origin. We can finally conclude that in order for the queen to search we must have $R(t) = \sqrt{(\IT-t)^2 - \sin^2\left(\frac{nt}{2}\right)} - \coss{\frac {nt} 2} \geq 1$.
\qed \end{proof}
Armed with these lemmas we are now able to tackle our main result.

\begin{proof}(Theorem~\ref{thm:lb})
    Set $\IT_0 = 1 + \frac{2}{n}\cos^{-1}\left(\frac{-2}{n}\right) + \sqrt{1 - \frac{4}{n^2}}$ and assume we have an algorithm with an evacuation time $\IT < \IT_0$. By Lemma~\ref{lm:lb2}, this implies that the queen must search the perimeter of $\UC$ before the time $t_c = 1 + \frac{2}{n}\cos^{-1}\left(\frac{-2}{n}\right)$.\footnote{Alternatively we can say that the robots must search at a collective 
    rate $> n$ by the time $t_c$. This is why we were able to ignore the ``unreasonable case'' in Lemma~\ref{lm:queen_loc}} Assume that at the time $t_c$ the robots have collectively searched the perimeter of $\UC$ at a rate $\mu$ satisfying $n < \mu \leq n+1$. Then at the time $t_c$ the unexplored subset of $\UC$ has length $\phi(t) = 2\pi-\mu (t_c-1) = 2\pi - 2\frac{\mu}{n}\acoss{\frac{-2}{n}} < \pi$. 
    Since $\phi(t_c) \leq \pi$ we can use Lemma~\ref{lm:distance} to say that the queen must be located within a distance of $R(t_c)$ of the origin at the time $t_c$. Furthermore, in order for the queen to have searched the perimeter of $\UC$ at the time $t_c$, we must have $R(t_c) \geq 1$. However, observe that
    \small 
    \begin{align*}
        R(t_c) &= \sqrt{(\IT-t_c)^2 - \sin^2\left(\frac{n(t_c-1)}{2}\right)} - \coss{\frac {n(t_c-1)} 2} \\
        &\leq \sqrt{(\IT_0-t_c)^2 - \sin^2\left(\frac{n(t_c-1)}{2}\right)} - \coss{\frac {n(t_c-1)} 2}\\
        &= \sqrt{1-\frac{4}{n^2} - \sin^2\left(\acoss{\frac{-2}{n}}\right)} - \coss{\acoss{\frac{-2}{n}}} \\
        &= \frac{2}{n}
    \end{align*}
    \normalsize
    which is clearly less than one for $n \geq 4$. We have therefore arrived to a contradiction and must conclude that the lower bound holds.
    
    To determine the asymptotic behaviour of $\IT_0$ we can compute a Taylor series of $\IT_0$ about $n = \infty$. We find that the first few terms in the series are $2+\frac{\pi}{n}+\frac{2}{n^2}$.
\qed \end{proof}

\section{Conclusions}
\label{secconclusion}

We studied an evacuation problem concerning priority search on the perimeter of a unit disk where only one robot (the queen) needs to exit from an unkown location. We focused on the case of $n \geq 4$ servants and showed in Section~\ref{sec6} that for any $n\geq4$ the queen can be evacuated in time at most $2+4(\sqrt{2}-1)\frac{\pi}{n}$. Furthermore, in Section~\ref{sec7}, we demonstrated that the queen cannot be evacuated in time less than $1+\frac{2}{n}\acoss{\frac{-2}{n}} + \sqrt{1-\frac{4}{n^2}} > 2+\frac{\pi}{n}+\frac {2}{n^2}$. Thus, in the limit of large $n$, we are left with a gap of $(4\sqrt{2}-5)\frac{\pi}{n} \approx 0.657\frac{\pi}{n}$ between the best upper and lower bounds. We conjecture that Algorithm~\ref{alg:ns} is in fact optimal. We will now justify this conjecture.

As was previously mentioned, one might think from Algorithm~\ref{alg:ns} that, since the queen is able to reach the perimeter of $\UC$ before the servants have finished their search, it would be possible to improve our algorithm. However, this is not the case -- similar to the proof of Theorem~\ref{thm:lb} there are critical times ($\frac n2$ of them) that occur before the queen reaches the perimeter and anything she does after these critical times cannot improve the evacuation time. These critical times result from a tradeoff between maximizing the rate at which the servants search -- for which the queen should remain near the origin -- and minimizing the distance of the queen from possible exits near the end of the algorithm -- for which the queen should be near the perimeter. Furthermore, in order to achieve the best tradeoff, the queen should travel as fast as she can from the origin to the perimeter. In other words, between these critical times, the queen should maximize her radial velocity. If we could prove that the queen does not need to participate in searching then it would not be so difficult to conclude why Algorithm~\ref{alg:ns} would be optimal. Any other trajectory of the queen between the critical search times will result in the same or a reduced radial velocity of the queen. It therefore does not seem likely that, with a reduced radial velocity, we can reduce the evacuation time.

In addition to improving the bounds obtained in this paper there are several interesting open problems related to priority search and evacuation. In particular, we may define a {\em weighted evacuation} problem (for a given group of agents) as a generalization of the priority evacuation problem studied here. One can differentiate on agent preferences by assigning a weight $w_i$ to each agent $i$ and require to evacuate a subset of agents of total weight $\geq W$ in minimum time. With this formulation in mind, the regular evacuation problem (see \cite{CGGKMP}) is the case where $w_i = 1$ for all agents and $W = n$, while for the problem considered in this work $w_i=0$ for all agents except the queen for which $w_{queen} = 1$ and $W = 1$.

\bibliographystyle{plain}
\bibliography{refs}

\begin{thebibliography}{10}

\bibitem{AH00}
S.~Albers and M.~R. Henzinger.
\newblock Exploring unknown environments.
\newblock {\em SIAM Journal on Computing}, 29(4):1164--1188, 2000.

\bibitem{alpern2002theory}
S.~Alpern and S.~Gal.
\newblock {\em The theory of search games and rendezvous}, volume~55.
\newblock Kluwer Academic Publishers, 2002.

\bibitem{baezayates1993searching}
R.~Baeza~Yates, J.~Culberson, and G.~Rawlins.
\newblock Searching in the plane.
\newblock {\em Information and Computation}, 106(2):234--252, 1993.

\bibitem{SIROCCO16}
E.~Bampas, J.~Czyzowicz, L.~Gasieniec, D.~Ilcinkas, R.~Klasing, T.~Kociumaka,
  and D.~Pajak.
\newblock Linear search by a pair of distinct-speed robots.
\newblock In {\em SIROCCO}, pages 195--211. Springer LNCS, 2016.

\bibitem{beck1964linear}
A.~Beck.
\newblock On the linear search problem.
\newblock {\em Israel J. of Mathematics}, 2(4):221--228, 1964.

\bibitem{bellman1963optimal}
R.~Bellman.
\newblock An optimal search.
\newblock {\em SIAM Review}, 5(3):274--274, 1963.

\bibitem{Bose16}
P.~Bose and J.-L. De~Carufel.
\newblock A general framework for searching on a line.
\newblock In {\em {WALCOM} 2016, Kathmandu, Nepal, March 29-31, 2016,
  Proceedings}, pages 143--153, 2016.

\bibitem{Bose13}
P.~Bose, J.-L. De~Carufel, and S.~Durocher.
\newblock Revisiting the problem of searching on a line.
\newblock In {\em {ESA} 2013, Sophia Antipolis, France, September 2-4, 2013.
  Proceedings}, pages 205--216, 2013.

\bibitem{Watten2017}
S.~Brandt, F.~Laufenberg, Y.~Lv, D.~Stolz, and R.~Wattenhofer.
\newblock Collaboration without communication: Evacuating two robots from a
  disk.
\newblock In {\em Algorithms and Complexity - 10th International Conference,
  {CIAC} 2017, Athens, Greece, May 24-26, 2017. Proceedings}, pages 104--115,
  2017.

\bibitem{Groupsearch}
M.~Chrobak, L.~Gasieniec, Gorry T., and R.~Martin.
\newblock Group search on the line.
\newblock In {\em SOFSEM 2015}, pages 164--176. Springer, 2015.

\bibitem{DBLP:conf/icdcn/CzyzowiczDGKM16}
J.~Czyzowicz, S.~Dobrev, K.~Georgiou, E.~Kranakis, and F.~MacQuarrie.
\newblock Evacuating two robots from multiple unknown exits in a circle.
\newblock In {\em Proceedings of the 17th International Conference on
  Distributed Computing and Networking, Singapore, January 4-7, 2016}, pages
  28:1--28:8, 2016.

\bibitem{CGGKMP}
J.~Czyzowicz, L.~Gasieniec, T.~Gorry, E.~Kranakis, R.~Martin, and D.~Pajak.
\newblock Evacuating robots from an unknown exit located on the perimeter of a
  disc.
\newblock In {\em DISC 2014}, pages 122--136. Springer, Austin, Texas, 2014.

\bibitem{georgioudiskfaulty2017}
J.~Czyzowicz, K.~Georgiou, M.~Godon, E.~Kranakis, D.~Krizanc, W.~Rytter, and
  M.~Wlodarczyk.
\newblock Evacuation from a disc in the presence of a faulty robot.
\newblock In {\em SIROCCO 2017, 19-22 June 2017, Porquerolles, France}, pages
  158--173, 2017.

\bibitem{CGKKKNOS}
J.~Czyzowicz, K.~Georgiou, R.~Killick, E.~Kranakis, K.~Krizanc, L.~Narayanan,
  J.~Opatrny, and S.~Shende.
\newblock God save the queen.
\newblock In {\em 9th International Conference on Fun With Algorithms
  (FUN’18)}, 2018.

\bibitem{DBLP:conf/ciac/CzyzowiczGKNOV15}
J.~Czyzowicz, K.~Georgiou, E.~Kranakis, L.~Narayanan, J.~Opatrny, and
  B.~Vogtenhuber.
\newblock Evacuating robots from a disk using face-to-face communication
  (extended abstract).
\newblock In {\em Algorithms and Complexity, {CIAC} 2015, Paris, France, May
  20-22, 2015. Proceedings}, pages 140--152, 2015.

\bibitem{DBLP:conf/adhoc-now/CzyzowiczKKNOS15}
J.~Czyzowicz, E.~Kranakis, D.~Krizanc, L.~Narayanan, J.~Opatrny, and S.~Shende.
\newblock Wireless autonomous robot evacuation from equilateral triangles and
  squares.
\newblock In {\em Ad-hoc, Mobile, and Wireless Networks, {ADHOC-NOW} 2015,
  Athens, Greece, June 29 - July 1, 2015, Proceedings}, pages 181--194, 2015.

\bibitem{demaine2006online}
E.~D. Demaine, S.~P. Fekete, and S.~Gal.
\newblock Online searching with turn cost.
\newblock {\em Theoretical Computer Science}, 361(2):342--355, 2006.

\bibitem{DKP91}
X.~Deng, T.~Kameda, and C.~Papadimitriou.
\newblock How to learn an unknown environment.
\newblock In {\em FOCS}, pages 298--303. IEEE, 1991.

\bibitem{feinerman2017fast}
Ofer Feinerman, Amos Korman, Shay Kutten, and Yoav Rodeh.
\newblock Fast rendezvous on a cycle by agents with different speeds.
\newblock {\em Theoretical Computer Science}, 688:77--85, 2017.

\bibitem{FT08}
F.~V. Fomin and D.~M. Thilikos.
\newblock An annotated bibliography on guaranteed graph searching.
\newblock {\em Theoretical Computer Science}, 399(3):236--245, 2008.

\bibitem{GeorgiouKK16}
Konstantinos Georgiou, George Karakostas, and Evangelos Kranakis.
\newblock Search-and-fetch with one robot on a disk - (track: Wireless and
  geometry).
\newblock In {\em Algorithms for Sensor Systems - 12th International Symposium
  on Algorithms and Experiments for Wireless Sensor Networks, {ALGOSENSORS}
  2016, Aarhus, Denmark, August 25-26, 2016, Revised Selected Papers}, pages
  80--94, 2016.

\bibitem{GeorgiouKK17}
Konstantinos Georgiou, George Karakostas, and Evangelos Kranakis.
\newblock Search-and-fetch with 2 robots on a disk - wireless and face-to-face
  communication models.
\newblock In Federico Liberatore, Greg~H. Parlier, and Marc Demange, editors,
  {\em Proceedings of the 6th International Conference on Operations Research
  and Enterprise Systems, ICORES 2017, Porto, Portugal, February 23-25, 2017},
  pages 15--26. SciTePress, 2017.

\bibitem{HIKK01}
F.~Hoffmann, C.~Icking, R.~Klein, and K.~Kriegel.
\newblock The polygon exploration problem.
\newblock {\em SIAM Journal on Computing}, 31(2):577--600, 2001.

\bibitem{lamprou2016fast}
I.~Lamprou, R.~Martin, and S.~Schewe.
\newblock Fast two-robot disk evacuation with wireless communication.
\newblock In {\em International Symposium on Distributed Computing}, pages
  1--15, 2016.

\bibitem{PY}
C.~H Papadimitriou and M.~Yannakakis.
\newblock Shortest paths without a map.
\newblock In {\em ICALP}, pages 610--620. Springer, 1989.

\bibitem{pattanayak2017evacuating}
D.~Pattanayak, H.~Ramesh, P.S. Mandal, and S.~Schmid.
\newblock Evacuating two robots from two unknown exits on the perimeter of a
  disk with wireless communication.
\newblock In {\em Proceedings of the 19th International Conference on
  Distributed Computing and Networking, {ICDCN} 2018, Varanasi, India, January
  4-7, 2018}, pages 20:1--20:4, 2018.

\bibitem{anim4}
Animation of algorithm~\ref{alg:ns} for $n=4$.
\newblock
  \url{https://drive.google.com/open?id=\\1OhmWeqFZLFLiwQalvPoZSTg9Ah860mMn}.
\newblock Feb. 13, 2018.

\bibitem{anim8}
Animation of algorithm~\ref{alg:ns} for $n=8$.
\newblock
  \url{https://drive.google.com/open?id=\\10ntWmekJr5pTywEfpTNAw6uyxxrfpHsA}.
\newblock Feb. 13, 2018.

\end{thebibliography}

\end{document}